\documentclass[journal,10pt,draftclsnofoot,onecolumn]{IEEEtran}
\usepackage[cmex10]{amsmath}
\usepackage{amssymb}
\usepackage{amscd}
\usepackage{textcomp}
\newtheorem{theorem}{Theorem}
\newtheorem{lemma}[theorem]{Lemma}
\newtheorem{corollary}[theorem]{Corollary}
\newtheorem{proposition}[theorem]{Proposition}

\newtheorem{example}[theorem]{Example}
\newtheorem{remark}[theorem]{Remark}

\ifCLASSINFOpdf
\else
\fi
\begin{document}
%
\title{Improved Semidefinite Programming Bound on Sizes of Codes}
%
%
%

\author{Hyun~Kwang~Kim,~\IEEEmembership{Member,~IEEE,}
        and~Phan~Thanh~Toan
\thanks{The authors are with the Department of Mathematics, Pohang University of Science and Technology, Pohang 790-784, Korea (e-mail: hkkim@postech.ac.kr; pttoan@postech.ac.kr).}
}

\maketitle

\begin{abstract}
Let $A(n,d)$ (respectively $A(n,d,w)$) be the maximum possible number of codewords in a binary code (respectively binary constant-weight $w$ code) of length $n$ and minimum Hamming distance at least $d$. By adding new linear constraints to Schrijver's semidefinite programming bound, which is obtained from block-diagonalising the Terwilliger algebra of the Hamming cube, we obtain two new upper bounds on $A(n,d)$, namely $A(18,8) \leq 71$ and $A(19,8) \leq 131$. Twenty three new upper bounds on $A(n,d,w)$ for $n \leq 28$ are also obtained by a similar way.
\end{abstract}



%
\IEEEpeerreviewmaketitle

\section{Introduction}
\IEEEPARstart{L}{et} $\mathcal F = \{0,1\}$ and let $n$ be a positive integer. The {\em (Hamming) distance} between two vectors in $\mathcal{F}^n$ is the number of coordinates where they differ. The {\em (Hamming) weight} of a vector in $\mathcal F^n$ is the distance between it and the zero vector. The {\em minimum distance} of a subset of $\mathcal F^n$ is the smallest distance between any two different vectors in that subset. An $(n,d)$ {\em code} is a subset of $\mathcal F^n$ having minimum distance $\geq d$. If $\mathcal C$ is an $(n,d)$ code, then an element of $\mathcal C$ is called a {\em codeword} and the number of codewords in $\mathcal C$ is called the {\em size} of $\mathcal C$.

The largest possible size of an $(n,d)$ code is denoted by $A(n,d)$. The problem of determining the exact values of $A(n,d)$ is one of the most fundamental problems in combinatorial coding theory. Among upper bounds on $A(n,d)$, Delsarte's linear programming bound is quite powerful (see \cite{d} and \cite{ms}). This bound is obtained from block-diagonalising the Bose-Mesner algebra of $\mathcal F^n$. In 2005, by block-diagonalising the Terwilliger algebra (which contains the Bose-Mesner algebra) of $\mathcal F^n$, Schrijver gave a semidefinite programming bound \cite{s}. This bound was shown to be stronger than or as good as Delsarte's linear programming bound. In fact, eleven new upper bounds on $A(n,d)$ were obtained in the paper for $n \leq 28$. In 2002, Mounits, Etzion, and Litsyn added more linear constraints to Delsarte's linear programming bound and obtained new upper bounds on $A(n,d)$ \cite{mel}. In this paper, we construct new linear constraints and show that these linear constraints improve Schrijver's semidefinite programming bound. Among improved upper bounds on $A(n,d)$ for $n \leq 28$, there are two new upper bounds, namely $A(18,8) \leq 71$ and $A(19,8) \leq 131$.

An $(n,d,w)$ {\em constant-weight code} is an $(n,d)$ code such that every codeword has weight $w$. Let $A(n,d,w)$ be the largest possible size of an $(n,d,w)$ constant-weight code. The problem of determining the exact values of $A(n,d,w)$ has its own interest. Upper bounds on $A(n,d,w)$ can even help to improve upper bounds on $A(n,d)$ (for example, see \cite{mel,ms}). There are also Delsarte's linear programming bound and Schrijver's semidefinite programming bound on $A(n,d,w)$ \cite{d,s}. In 2000, Agrell, Vardy, and Zeger added new linear constraints to Delsarte's linear programming bound and improved several upper bounds on $A(n,d,w)$ \cite{avz}. More linear constraints that improve upper bounds on $A(n,d,w)$ can be found in \cite{kkt}. In this paper, we add further new linear constraints to Schrijver's semidefinite programming bound on $A(n,d,w)$ and obtain twenty three new upper bounds on $A(n,d,w)$ for $n \leq 28$.

%
%
%
%



\section{Upper Bounds on $A(n,d)$}

In this section, we improve upper bounds on $A(n,d)$ by adding more linear constraints to Schrijver's semidefinite programming bound, which is obtained from block-diagonalising the Terwilliger algebra of the Hamming cube $\mathcal F^n$. For more details about Schrijver's semidefinite programming bound, see \cite{s}.

\subsection{General Definition of $A(n,d)$ and $A(n,d,w)$}

We first give a general definition. Let $n$ and $d$ be positive integers. For a finite (possibly empty) set $\Lambda = \{(X_i,d_i)\}_{i \in I}$, where each $X_i$ is a vector in $\mathcal F^n$ and each $d_i$ is a nonnegative integer, we define
\begin{eqnarray}
\nonumber A(n,\Lambda,d) & = &\mbox{maximum possible number of}\\
\nonumber & & \mbox{codewords in a binary code of}\\
\nonumber & & \mbox{length } n \mbox{ and minimum distance}\\
\nonumber & & \geq d \mbox{ such that each codeword is}\\
& & \mbox{at distance } d_i \mbox{ from } X_i, \forall i \in I.
\end{eqnarray}

\subsubsection{$|\Lambda|=0$}
If $\Lambda$ is empty, then we get the usual definition of $A(n,d)$.

\subsubsection{$|\Lambda|=1$}
If $\Lambda$ contains only one element, says $(X_1,d_1)$, then $A(n,\Lambda,d)$ is the maximum possible number of codewords in a binary code of length $n$ and minimum distance $\geq d$ such that each codeword is at distance $d_1$ from $X_1$. By translation, we may assume that $X_1$ is the zero vector so that each codeword has weight $d_1$. Therefore,
\begin{eqnarray}
A(n,\Lambda,d) = A(n,d,w),
\end{eqnarray}
where $w=d_1$.

A $(w_1, n_1, w_2, n_2, d)$ {\it doubly-constant-weight code} is an $(n_1 + n_2, d, w_1 + w_2)$ constant-weight code such that every codeword has exactly $w_1$ ones on the first $n_1$ coordinates (and hence has exactly $w_2$ ones on the last $n_2$ coordinates). Let $T(w_1, n_1, w_2, n_2, d)$ be the largest possible size of a $(w_1, n_1, w_2, n_2, d)$ doubly-constant-weight code. Agrell, Vardy, and Zeger showed in \cite{avz} that upper bounds on $T(w_1, n_1, w_2, n_2, d)$ can help improving upper bounds on $A(n,d,w)$. In our result, upper bounds on $T(w_1, n_1, w_2, n_2, d)$ will be used to improve upper bounds on $A(n,d)$. As $A(n,d)$ and $A(n,d,w)$, $T(w_1, n_1, w_2, n_2, d)$ is also a special case of $A(n,\Lambda,d)$.

\subsubsection{$|\Lambda|=2$} If $\Lambda$ contains two elements, then the following proposition shows that $A(n,\Lambda,d)$ is exactly $T(w_1, n_1, w_2, n_2, d)$.

\begin{proposition} \label{mainp}
If $\Lambda = \{(X_1,d_1),(X_2,d_2)\}$, then
\begin{eqnarray}
A(n,\Lambda,d) = T(w_1, n_1, w_2, n_2, d),
\end{eqnarray}
where $n_1 =d(X_1,X_2), n_2 = n - n_1, w_1= \frac{1}{2}(d_1 - d_2 + n_1)$, and $w_2 = \frac{1}{2}(d_1 + d_2 - n_1).$
\end{proposition}

\begin{proof}
Let $n_1 = d(X_1, X_2)$ and $n_2 = n - n_1$. By translation, we may assume that $X_1$ is the zero vector. Hence, $d(X_1,X_2) = wt(X_2)$. Let $Y$ be a vector at distance $d_1$ from $X_1$ and at distance $d_2$ from $X_2$. By rearranging the coordinates, we may assume that
\begin{center}
$\begin{array}{cccc}
X_1 & = & \overbrace{0 \cdots 0 0 \cdots 0}^{n_1} & \overbrace{0 \cdots 0 0 \cdots 0}^{n_2}\\
X_2 & = & 1 \cdots 1 1 \cdots 1 & 0 \cdots 0 0 \cdots 0\\
Y   & = & 0 \cdots 0 \underbrace{1 \cdots 1}_{w_1} & \underbrace{1 \cdots 1}_{w_2} 0 \cdots 0
\end{array}.$
\end{center}
Since $X_1$ is the zero vector, we have
\begin{eqnarray}\label{four}
w_1 + w_2 = wt(Y) = d(Y,X_1) = d_1.
\end{eqnarray}
Also,
\begin{eqnarray}\label{five}
(n_1 - w_1) + w_2 = d(Y,X_2) = d_2.
\end{eqnarray}
(\ref{four}) and (\ref{five}) give $w_1 = \frac{1}{2}(d_1 - d_2 + n_1)$ and $w_2 = \frac{1}{2} (d_1 + d_2 - n_1)$.
\end{proof}

\subsubsection{$|\Lambda|\geq 3$}

It becomes more complicated when $\Lambda$ contains more than two elements. We consider a very special case when $|\Lambda|=4$, which will be used in our improving upper bounds on $A(n,d,w)$ in Section \ref{sba}. Suppose that $\Lambda =\{(X_1, d_1), (X_2,d_2), (X_3,d_3), (X_4, d_4)\}$ satisfies the following conditions.
\begin{itemize}
\item $X_1$ is the zero vector (which can always be assumed).
\item $X_2$ and $X_3$ have the same weight $d_1$.
\item $X_4 = X_2 + X_3$.
\end{itemize}
Then $A(n, \Lambda, d) = T(w_1, n_1, w_2, n_2, w_3, n_3, w_4, n_4, d)$, where $w_i$ and $n_i$ $(1 \leq i \leq 4)$ are determined in the next proposition. The definition of $T(w_1, n_1, w_2, n_2, w_3, n_3, w_4, n_4, d)$ is similar to that of $T(n_1,w_1,n_2,w_2,d)$ (it is the largest possible size of a $(\sum_{i=1}^4 n_i, d)$ code such that on each codeword there are exactly $w_i$ ones on the $n_i$ coordinates $(1 \leq i \leq 4)$).

\begin{proposition}\label{mainp2}
Suppose that $\Lambda =\{(X_i, d_i)\}_{i=1}^4$ satisfies $X_1$ is the zero vector, $wt(X_2) = wt(X_3) = d_1$, and $X_4 = X_2 + X_3$. Then
\begin{eqnarray}
A(n, \Lambda, d) = T(w_1, n_1, w_2, n_2, w_3, n_3, w_4, n_4, d),
\end{eqnarray}
where
$n_1 = n_3 = \frac{1}{2}d(X_2,X_3), n_2 = d_1 - n_1, n_4 = n-n_1-n_2-n_3,$
\begin{eqnarray}
\nonumber w_1 &=& \frac{1}{4} (d_1 - d_2 + d_3 - d_4) + \frac{1}{2} n_1,\\
\nonumber w_2 &=& \frac{1}{4} (d_1 - d_2 - d_3 + d_4) + \frac{1}{2} n_2,\\
\nonumber w_3 &=& \frac{1}{4} (d_1 + d_2 - d_3 - d_4) + \frac{1}{2} n_3,\\
\nonumber w_4 &=& \frac{1}{4} (d_1 + d_2 + d_3 + d_4) + \frac{1}{2} (n_4-n).
\end{eqnarray}
\end{proposition}

\begin{proof}
Suppose that $Z$ is a vector at distance $d_i$ from $X_i$ $(1 \leq i \leq 4)$. By rearranging the coordinates, we may assume the following.
\begin{center}
$\begin{array}{rl}
X_2 = \overbrace{1 \cdots \cdots \cdots 1}^{n_1} \overbrace{1 \cdots \cdots \cdots 1 }^{n_2} & \overbrace{0 \cdots \cdots \cdots 0}^{n_3} \overbrace{0 \cdots \cdots \cdots 0}^{n_4}\\
X_3  = 0 \cdots \cdots \cdots 0 ~\!1 \cdots \cdots \cdots 1 &1 \cdots \cdots \cdots 1 ~\! 0 \cdots \cdots \cdots 0\\
Z = 0 \cdots 0 \underbrace{1 \cdots 1}_{w_1}  \underbrace{1 \cdots 1}_{w_2} 0 \cdots 0 & 0 \cdots 0  \underbrace{1 \cdots 1}_{w_3} \underbrace{1 \cdots 1}_{w_4} 0 \cdots 0
\end{array}$
\end{center}
Let $n_1, n_2, n_3, n_4$ be as in the above figure. Since $n_1 + n_3 = d(X_2, X_3)$ and $X_2, X_3$ have the same weight, $n_1 = n_3 = \frac{1}{2} d(X_2, X_3)$. Now $n_1 + n_2 = wt(X_2) = d_1$. Therefore, $n_2 = d_1 - n_1$ and $n_4 = n- n_1 - n_2 - n_3$. We have
\begin{center}
$\left\{\begin{array}{lcl}
w_1 + w_2 + w_3 + w_4 = wt(Z) &=& d(Z,X_1) = d_1\\
(n_1-w_1) + (n_2-w_2) + w_3 + w_4 &=& d(Z,X_2) = d_2\\
w_1 + (n_2 - w_2) + (n_3 - w_3) + w_4 &=& d(Z,X_3) = d_3\\
(n_1-w_1) + w_2 + (n_3 - w_3) + w_4 &=& d(Z,X_4) = d_4
\end{array} \right..$
\end{center}
Solving these equations, we get $w_i$ $(1 \leq i \leq 4)$ as desired.
\end{proof}

\subsection{Schrijver's Semidefinite Programming Bound on $A(n,d)$}

Let $\mathcal P$ be the collection of all subsets of $\{1, 2, \ldots, n\}$. Each vector in $\mathcal F^n$ can be identified with its support (the support of a vector is the set of coordinates at which the vector has nonzero entries). With this identification, a code is a subset of $\mathcal P$ and the (Hamming) distance between two subsets $X$ and $Y$ in $\mathcal P$ is $d(X,Y) = |X \Delta Y|$.
Let $\mathcal C$ be an $(n,d)$ code. For each $i, j,$ and $t$, define
\begin{eqnarray}\label{xtij}
x^t_{i,j} = \frac{1}{|\mathcal C| \left( n \atop i-t,j-t,t \right)} \lambda^t_{i,j},
\end{eqnarray}
where $\left(a \atop b_1, b_2, \ldots, b_m \right)$ denotes the number of pairwise disjoint subsets of sizes $b_1, b_2, \ldots, b_m$ respectively of a set of size $a$, and $\lambda^t_{i,j}$ denotes the number of triples $(X,Y,Z) \in \mathcal C^3$ with $|X\Delta Y| = i, |X \Delta Z| = j,$ and $|(X \Delta Y) \cap (X \Delta Z)| = t$, or equivalently, with $|X \Delta Y| = i$, $|X \Delta Z| = j$, and $|Y \Delta Z| = i+j-2t$. Set $x^t_{i,j}=0$ if $\left( n \atop i-t,j-t,t \right) = 0$.

The key part of Schrijver's semidefinite programming bound is that for each $k = 0, 1, \ldots, \lfloor \frac{n}{2} \rfloor$, the matrices
\begin{eqnarray}\label{mt1}
\left( \sum_{t=0}^n \beta^t_{i,j,k} x^t_{i,j} \right)_{i,j = k}^{n-k}
\end{eqnarray}
and
\begin{eqnarray}\label{mt2}
\left(\sum_{t=0}^n \beta^t_{i,j,k}(x^0_{i+j-2t,0} - x^t_{i,j}) \right)_{i,j=k}^{n-k}
\end{eqnarray}
are positive semidefinite, where $\beta^t_{i,j,k}$ is given by
\begin{eqnarray}
\beta^t_{i,j,k} & = & \sum_{u=0}^n (-1)^{u-t}\left(u \atop t \right) \left( n - 2k \atop u-k \right) \left(n-k-u \atop i-u\right) \left(n-k-u \atop j-u \right).
\end{eqnarray}
Since
\begin{eqnarray}\label{tong}
|\mathcal C| = \sum_{i=0}^n \left(n \atop i \right) x^0_{i,0},
\end{eqnarray}
an upper bound on $A(n,d)$ can be obtained by considering the $x^t_{i,j}$ as variables and by
\begin{eqnarray} \label{max}
\mbox{maximizing } \sum_{i=0}^n \left(n \atop i \right) x^0_{i,0}
\end{eqnarray}
subject to the matrices (\ref{mt1}) and (\ref{mt2}) are positive semidefinite for each $k = 0, 1, \ldots, \lfloor \frac{n}{2} \rfloor$ and subject to the following conditions on the $x^t_{i,j}$ (see \cite{s}).

\begin{itemize}
\item [(i)] $x^0_{0,0} = 1$.
\item [(ii)] $0 \leq x^t_{i,j} \leq x^0_{i,0}$ \mbox{ and } $x^0_{i,0} + x^0_{j,0} \leq 1 + x^t_{i,j}$ for all $i,j,t \in \{0,1, \ldots, n\}$.
\item [(iii)] $x^t_{i,j} = x^{t'}_{i',j'}$ if $(i',j',i'+j'-2t')$ is a permutation of $(i,j,i+j-2t)$.
\item [(iv)] $x^t_{i,j} = 0$ if $\{i, j, i+j-2t\} \cap \{1, 2, \ldots, d-1\} \not = \emptyset$.
\end{itemize}

\subsection{Improved Schrijver's Semidefinite Programming Bound on $A(n,d)$}

\subsubsection{New Constraints for $x^t_{i,j}$}
Let $\mathcal C$ be an $(n,d)$ code and let $x^t_{i,j}$ be defined by (\ref{xtij}).
\begin{theorem}\label{kq1}
For all $i, j, t \in \{0, 1, \ldots, n\}$ with $\left(n \atop i-t,j-t,t\right) \not = 0$,
\begin{eqnarray}
x^t_{i,j} \leq \frac{T(t,i,j-t,n-i,d)}{\left( i \atop t \right) \left(n-i \atop j-t \right)} x^0_{i,0}.
\end{eqnarray}
\end{theorem}

\begin{proof}
Recall that $\lambda^t_{i,j}$ is the number of triples $(X,Y,Z) \in \mathcal C^3$ with $|X \Delta Y| = i$, $|X \Delta Z| = j$, and $|Y \Delta Z| = i+j-2t$. For any pair $(X,Y) \in \mathcal C^2$ with $|X \Delta Y|= i$, the number of $Z \in \mathcal C$ such that $|Z \Delta X|=j$ and $|Z \Delta Y| = i+j-2t$ is upper bounded by $A(n,\Lambda,d)$, where $\Lambda = \{(X,j),(Y,i+j-2t)\}$. By Proposition \ref{mainp},
\begin{eqnarray}
A(n,\Lambda,d) = T(t,i,j-t,n-i,d).
\end{eqnarray}
Since the number of pairs $(X,Y) \in \mathcal C^2$ such that $|X \Delta Y| = i$ is $\lambda^0_{i,0}$,
\begin{eqnarray}
\lambda^t_{i,j} \leq T(t,i,j-t,n-i,d) \lambda^0_{i,0}.
\end{eqnarray}
Therefore,
\begin{eqnarray}
\nonumber x^t_{i,j} & = & \frac{1}{|\mathcal C| \left(n \atop i-t, j-t, t\right)} \lambda^t_{i,j} \\
\nonumber &\leq&  \frac{T(t,i,j-t,n-i,d)}{|\mathcal C| \left(n \atop i-t, j-t, t\right)} \lambda^0_{i,0}\\
\nonumber & = & \frac{T(t,i,j-t,n-i,d) \left(n \atop i \right)}{\left(n \atop i-t, j-t, t\right)} x^0_{i,0} \\
\nonumber & = & \frac{T(t,i,j-t,n-i,d)}{\left(i \atop t \right) \left(n-i \atop j-t \right)} x^0_{i,0}.
\end{eqnarray}
\end{proof}

The following corollary was used in \cite{s}.

\begin{corollary}\label{hq3}
For each $j \in \{0, 1, \ldots, n\}$,
\begin{eqnarray}
{\left(n \atop j \right)}x^0_{0,j} \leq {A(n,d,j)}.
\end{eqnarray}
\end{corollary}

\begin{proof}
By Theorem \ref{kq1}, we have
\begin{eqnarray}
x^0_{0,j} \leq \frac{T(0,0,j,n,d)}{\left(0 \atop 0\right)\left(n \atop j\right)}x^0_{0,0} = \frac{A(n,d,j)}{\left(n \atop j \right)}.
\end{eqnarray}
\end{proof}

\begin{remark}
Theorem \ref{kq1} improve the condition $x^{t}_{i,j} \leq x^0_{i,0}$ in Schrijver's semidefinite programming bound since $\frac{T(t,i,j-t,n-i,d)}{\left(i \atop t \right) \left(n-i \atop j-t\right)} \leq 1$ (in fact, $\frac{T(t,i,j-t,n-i,d)}{\left(i \atop t \right) \left(n-i \atop j-t\right)}$ is much less than $1$ in general). Similarly, Corollary \ref{hq3} in many cases (of $i$ and $j$) improve the condition $x^0_{i,0} + x^0_{j,0} \leq 1 + x^t_{i,j}$ since $x^0_{u,0} = x^0_{0,u} = \frac{A(n,d,u)}{\left( n \atop u \right)}$ is much less than $\frac{1}{2}$ in general.
\end{remark}

\subsubsection{Delsarte's Linear Programming Bound and Its Improvements}

Let $\mathcal C$ be an $(n,d)$ code, the {\em distance distribution} $\{B_i\}_{i=0}^n$ of $\mathcal C$ is defined by
\begin{eqnarray}
B_i = \frac{1}{|\mathcal C|} \cdot |\{(X,Y) \in \mathcal C^2 \mid |X \Delta Y| = i\}|.
\end{eqnarray}
By definition,
\begin{eqnarray}
\left(n \atop i \right) x^0_{i,0} = B_i
\end{eqnarray}
for each $i = 0, 1, \ldots, n$. Hence, $\{\left(n \atop i \right) x^0_{i,0}\}_{i=0}^n$ is the distance distribution on $\mathcal C$. The following result can be found for example in \cite{bbm} or \cite{kkt}.

\begin{theorem}\label{del} (Delsarte's linear programming bound and its improvements).
Let $\mathcal C$ be an $(n,d)$ code with distance distribution
$\{B_i\}_{i=0}^n = \{\left(n \atop i \right) x^0_{i,0}\}_{i=0}^n$. For $k = 1,2, \ldots, n$,
\begin{eqnarray}\label{thchan}
\sum_{i=1}^n P_k(n;i) B_i  \geq - \left(n \atop k \right),
\end{eqnarray}
where $P_k(n;x)$ is the Krawtchouk polynomial given by
\begin{eqnarray}
P_k(n;x) = \sum_{j=0}^k (-1)^j \left(x \atop j \right) \left( n- x \atop k-j \right).
\end{eqnarray}
If $M = |\mathcal C|$ is odd, then
\begin{eqnarray}\label{thle}
\sum_{i=1}^n P_k(n;i) B_i \geq -\left(n \atop k \right) + \frac{1}{M}\left(n \atop k \right).
\end{eqnarray}
If $M = |\mathcal C| \equiv 2$ (mod $4$), then there exists $t \in \{0, 1, \ldots, n\}$ such that
\begin{eqnarray}\label{th3}
\sum_{i=1}^n P_k(n;i) B_i  \geq -\left(n \atop k \right)+ \frac{2}{M} \left[ \left(n \atop k \right) + P_k(n;t) \right].
\end{eqnarray}

\end{theorem}

\subsubsection{Linear Constraints on Distance Distributions $\{B_i\}_{i=0}^n$}

If some linear constraints are used to improve Delsarte's linear programming bound on $A(n,d)$, then these constraints can still be added to Schrijver's semidefinite programming bound to improve upper bounds on $A(n,d)$. The following constraints are due to Mounits, Etzion, and Litsyn (see \cite[Theorems 9 and 10]{mel}).

\begin{theorem}\label{kq2}
Let $\mathcal C$ be an $(n,d)$ code with distance distribution $\{B_i\}_{i=0}^n$. Suppose that $d$ is even and $\delta=d/2$. Then
\begin{eqnarray}
B_{n-\delta} + \left\lfloor \frac{n}{\delta}\right\rfloor \sum_{i<\delta} B_{n-i} \leq \left\lfloor \frac{n}{\delta} \right\rfloor
\end{eqnarray}
and
\begin{eqnarray}
B_{n-\delta-i} + [A(n,d,\delta+i) - A(n-\delta+i,d,\delta+i)] B_{n-\delta+i} + A(n,d,\delta+i) \sum_{j>i} B_{n-\delta+j} \leq A(n,d,\delta+i)
\end{eqnarray}
for all $i = 1, 2, \ldots, \delta-1$.
\end{theorem}

\begin{table}[!t]
\renewcommand{\arraystretch}{1.3}
\caption{Improved upper bounds for $A(n,d)$}
\label{table_example}
\centering
\begin{tabular}{|r|r|r|r|r|r|r|}
\hline
  &   & \mbox{best}  & \mbox{best upper} &              &                   & \mbox{} \\
  &   & \mbox{lower} & \mbox{bound}      & \mbox{new}   & \mbox{improved}   & \mbox{} \\
  &   & \mbox{bound} & \mbox{previously} & \mbox{upper} & \mbox{Schijver}   & \mbox{Schrijver} \\
n & d & \mbox{known} & \mbox{known}      & \mbox{bound} & \mbox{bound}      & \mbox{bound} \\
\hline
\hline
18 & 8 & 64 & 72 & 71 & 71 & 80\\
19 & 8 & 128 & 135 & 131 & 131& 142\\
20 & 8 & 256 & 256 & & 262 & 274\\
25 & 8 & 4096 & 5421 & & 5465 & 5477\\
26 & 8 & 4104 & 9275 & & 9649 & 9697\\
\hline
26 & 10 & 384 & 836 & & 885 & 886 \\
\hline
25 & 12 & 52 & 55 &  & 57 & 58 \\
26 & 12 & 64 & 96 &  & 97 & 98\\
\hline
\end{tabular}
\end{table}

Table \ref{table_example} shows improved upper bounds on $A(n,d)$ when linear constraints in Theorems \ref{kq1}, \ref{del}, and \ref{kq2} are added to Schrijver's semidefinite programming bound (\ref{max}). In the table, by Schrijver bound we mean upper bound obtained from Schrijver's semidefinite programming bound (\ref{max}). Among improved upper bounds on $A(n,d)$, there are two new upper bounds, namely
$$A(18,8) \leq 71 \quad \mbox{ and } \quad A(19,8) \leq 131.$$ The other best known upper bounds are from \cite{gms}. As in \cite{s}, all computations here were done by the algorithm SDPT3 available online on the NEOS Server for Optimization (http://www.neos-server.org/neos/solvers/index.html).

\begin{remark}
Since $A(n,d) = A(n+1,d+1)$ if $d$ is odd, we can always assume that $d$ is even. If $d$ is even, then $A(n,d)$ is attained by a code with all codewords having even weights. Hence, in Schrijver's semidefinite programming bound, one can put $x^t_{i,j} = 0$ if $i$ or $j$ is odd.
\end{remark}

\begin{remark}
In Theorems \ref{kq1} and \ref{kq2}, the values of $A(n,d,w)$ and $T(w_1,n_1,w_2,n_2,d)$ may have not yet been known. However, we can replace them by any of their upper bounds (see the proof of \cite[Theorem 10]{mel} for the validity of this replacement in Theorem \ref{kq2}). While best known upper bounds on $A(n,d,w)$ (which are mostly from \cite{bsss,avz,s,o}) are used in our computations, all upper bounds on $T(w_1, n_1, w_2, n_2,d)$ that we used are from the tables on Erik Agrell's website http://webfiles.portal.chalmers.se/s2/research/kit/bounds/dcw.html.
\end{remark}

\section{Upper bounds on $A(n,d,w)$}\label{sba}

\subsection{Some Properties of $A(n,d,w)$}
We begin with some elementary properties of $A(n,d,w)$ which can be found in \cite{ms}.

\begin{theorem}
\begin{eqnarray}\label{a}
A(n,d,w) = A(n,d+1,w), \quad \mbox{ if } d \mbox{ is odd,}
\end{eqnarray}
\begin{eqnarray}\label{b}
A(n,d,w) = A(n,d,n-w),
\end{eqnarray}
\begin{eqnarray}\label{c}
A(n,2,w) = \left(n \atop w \right),
\end{eqnarray}
\begin{eqnarray}\label{d}
A(n,2w,w) = \left\lfloor \frac{n}{w} \right\rfloor,
\end{eqnarray}
\begin{eqnarray}\label{e}
A(n,d,w) = 1, \quad \mbox{ if } 2w < d.
\end{eqnarray}
\end{theorem}

\begin{remark}
By (\ref{a}) and (\ref{c}), we can always assume that $d$ is even and $d \geq 4$. Also, by (\ref{b}), (\ref{d}), and (\ref{e}), we can assume that $d < 2w \leq n$.
\end{remark}

\subsection{Schrijver's Semidefinite Programming Bound on $A(n,d,w)$}

Let $\mathcal C$ be an $(n,d,w)$ constant-weight code and let $v = n-w$. For each $t$, $s$, $i$, and $j$, define
\begin{eqnarray}\label{defofy}
y^{t,s}_{i,j} = \frac{1}{|\mathcal C| \left( w \atop i-t, j-t, t \right) \left(v \atop i-s, j-s, s \right)} \mu^{t,s}_{i,j},
\end{eqnarray}
where $\mu^{t,s}_{i,j}$ is the number of triples $(X,Y,Z) \in \mathcal C^3$ with $|X \setminus Y| = i, |X \setminus Z| = j, |(X \setminus Y) \cap (X \setminus Z)| = t$, and $|(Y \setminus X) \cap (Z \setminus X)| = s$, or equivalently, with $|X \Delta Y| = 2i$, $|X \Delta Z| = 2j$, $|Y \Delta Z| = 2(i+j-t-s)$, and $|X \Delta Y \Delta Z| = w + 2t - 2s$. Set $y^{t,s}_{i,j} = 0$ if either $\left( w \atop i-t, j-t, t \right) = 0$ or $\left(v \atop i-s, j-s, s \right) = 0$.

In the previous section, $\beta^t_{i,j,k}$ depends on $n$. Hence, $\beta^t_{i,j,k}$ should be denoted by $\beta^{t,n}_{i,j,k}$. We will use the later notation in this section. As in \cite{s}, for each $k = 0, 1, \ldots, \lfloor \frac{w}{2} \rfloor$ and each $l = 0, 1, \ldots, \lfloor \frac{v}{2} \rfloor$, the matrices
\begin{eqnarray}\label{rr}
\left( \sum_{t,s} \beta^{t,w}_{i,j,k} \beta^{s,v}_{i,j,l} y^{t,s}_{i,j} \right)_{i,j \in W_k \cap V_l}
\end{eqnarray}
and
\begin{eqnarray}\label{rrprime}
\left( \sum_{t,s} \beta^{t,w}_{i,j,k} \beta^{s,v}_{i,j,l} (y^{0,0}_{i+j-t-s,0} - y^{t,s}_{i,j}) \right)_{i,j \in W_k \cap V_l}
\end{eqnarray}
are positive semidefinite, where $W_k = \{k, k+1, \ldots, w-k\}$ and $V_l = \{l, l+1, \ldots, v-l\}$.
Since
\begin{eqnarray}
|\mathcal C| = \sum_{i=0}^{\min\{w,v\}} \left(w \atop i \right) \left(v \atop i \right) y^{0,0}_{i,0},
\end{eqnarray}
an upper bound on $A(n,d,w)$ can be obtained by considering the $y^{t,s}_{i,j}$ as variables and by
\begin{eqnarray}\label{maxcwc}
\mbox{maximizing } \sum_{i=0}^{\min\{w,v\}} \left(w \atop i \right) \left(v \atop i \right) y^{0,0}_{i,0}
\end{eqnarray}
subject to the matrices (\ref{rr}) and (\ref{rrprime}) are positive semidefinite for each $k = 0, 1, \ldots, \lfloor \frac{w}{2} \rfloor$ and each $l = 0, 1, \ldots, \lfloor \frac{v}{2} \rfloor$, and subject to the following conditions.
\begin{itemize}
\item [(i)] $y^{0,0}_{0,0} = 1$.
\item [(ii)] $0 \leq y^{t,s}_{i,j} \leq y^{0,0}_{i,0}$ and $y^{0,0}_{i,0} + y^{0,0}_{j,0} \leq 1 + y^{t,s}_{i,j}$ for all $i, j, t, s \in \{0, 1, \ldots, \min\{w, v \}\}$.
\item [(iii)] $y^{t,s}_{i,j} = y^{t',s'}_{i',j'}$ if $t'-s' = t-s$ and $(i',j',i'+j'-t'-s')$ is a permutation of $(i,j,i+j-t-s)$.
\item [(iv)] $y^{t,s}_{i,j} = 0$ if $\{2i,2j, 2(i+j-t-s)\} \cap \{1, 2, \ldots, d-1\} \not = \emptyset$.
\end{itemize}

\subsection{Improved Schrijver's Semidefinite Programming Bound on $A(n,d,w)$}

\subsubsection{New Constraints for $y^{t,s}_{i,j}$}

Let $\mathcal C$ be an $(n,d,w)$ constant-weight code and let $y^{t,s}_{i,j}$ be defined by (\ref{defofy}). The following theorem corresponds to Theorem \ref{kq1} in the previous section.
\begin{theorem}\label{tr1}
For all $i,j,s,t \in \{0, 1, \ldots, \min\{w, v\}\}$ with $\left(w \atop i-t, j-t,t \right) \not = 0$ and $\left(v \atop i-s,j-s,s\right) \not =0$,
\begin{eqnarray}
y^{t,s}_{i,j} \leq \frac{T(t,i,j-t,w-i,s,i,j-s,v-i,d)}{\left(i \atop t \right) \left(w-i \atop j-t \right) \left(i \atop s \right) \left(v-i \atop j-s\right)} y^{0,0}_{i,0}.
\end{eqnarray}
\end{theorem}

\begin{proof}
Suppose that $(X,Y) \in \mathcal C^2$ such that $|X \Delta Y| = 2i$. We claim that the number of codewords $Z \in \mathcal C$ such that $|X \Delta Z| = 2j$, $|Y \Delta Z| = 2(i+j-t-s)$, and $|X \Delta Y \Delta Z| = w + 2t - 2s$ is upper bounded by $T(t,i,j-t,w-i, s,i,j-s,v-i, d)$. It is easy to see that this number is upper bounded by $A(n,\Lambda,d)$, where $\Lambda = \{(0,w),(X,2j), (Y,2(i+j-t-s)),(X \Delta Y, w+2t-2s)\}$. By Proposition \ref{mainp2},
\begin{eqnarray}
A(n,\Lambda,d) = T(w_1, n_1, w_2, n_2, w_3, n_3, w_4, n_4, d),
\end{eqnarray}
where $n_1 = n_3 = \frac{1}{2} |X \Delta Y| = i$, $n_2 = d_1 - n_1 = w - i$, $n_4 = n - i - (w-i) - i = v-i$, and similarly, $w_1 = i-t, w_2 = (w-i)-(j-t), w_3 = s, w_4 = j-s$. Hence,
\begin{eqnarray}
\nonumber A(n,\Lambda,d) & = & T(i-t,i,(w-i)-(j-t),w-i,s, i, j-s, v-i, d)\\
& = & T(t, i, j-t, w-i, s, i, j-s, v-i, d),
\end{eqnarray}
where the later equality comes from Proposition \ref{append} (iii) in the appendix. Since the number of pairs $(X,Y) \in \mathcal C^2$ such that $|X \Delta Y|=2i$ is $\mu^{0,0}_{i,0}$,
\begin{eqnarray}
\mu^{t,s}_{i,j} \leq T(t, i, j-t, w-i, s, i, j-s, v-i, d) \mu^{0,0}_{i,0}.
\end{eqnarray}
Therefore,
\begin{eqnarray}
\nonumber y^{t,s}_{i,j} & = & \frac{1}{|\mathcal C| \left(w \atop i-t, j-t, t \right) \left( v \atop i-s, j-s, s\right)} \mu^{t,s}_{i,j}\\
\nonumber & \leq & \frac{T(t, i, j-t, w-i, s, i, j-s, v-i, d)}{|\mathcal C| \left(w \atop i-t, j-t, t \right) \left( v \atop i-s, j-s, s\right)} \mu^{0,0}_{i,0}\\
\nonumber & = & \frac{T(t, i, j-t, w-i, s, i, j-s, v-i, d)}{\left(w \atop i-t, j-t, t \right) \left( v \atop i-s, j-s, s\right) \left(w \atop i \right)^{-1} \left( v \atop i \right)^{-1}}  y^{0,0}_{i,0}\\
\nonumber & = & \frac{T(t, i, j-t, w-i, s, i, j-s, v-i, d)}{\left(i \atop t\right) \left(w-i \atop j-t\right) \left(i \atop s \right) \left(v-i \atop j-s \right)}  y^{0,0}_{i,0}.
\end{eqnarray}
\end{proof}

\subsubsection{Delsarte's Linear Programming Bound}

Let $\mathcal C$ be an $(n,d,w)$ constant-weight code with distance distribution $\{B_i\}_{i=0}^n$. By definition of $y^{t,s}_{i,j}$,
\begin{eqnarray}
\left(w \atop i \right) \left(v \atop i \right) y^{0,0}_{i,0} = B_{2i}
\end{eqnarray}
for every $i$ (note that $B_0=1$ and $B_i = 0$ whenever $i$ is odd or $0 < i < d$ or $i > 2w$).

\begin{theorem}\label{tr2} (Delsarte's linear programming bound).
If $\{B_{i}\}_{i=0}^n$ is the distance distribution of an $(n,d,w)$ constant-weight code, then for $k = 1, 2, \ldots, w$,
\begin{eqnarray}
\sum_{i=d/2}^w q(k,i,n,w)B_{2i} \geq -1,
\end{eqnarray}
where
\begin{eqnarray}
q(k,i,n,w) = \frac{\sum_{j=0}^i (-1)^j \left( k \atop j \right) \left(w-k \atop i-j \right) \left( n-w-k \atop i - j \right)}{\left(w \atop i \right) \left(n-w \atop i \right)}.
\end{eqnarray}
\end{theorem}

Specifying Delsarte's linear programming bound on $A(n,d)$ gives the following linear constraints on $B_i$, which sometimes help reducing upper bounds on $A(n,d,w)$ by $1$ (see \cite[Proposition 11]{kkt}).

\begin{theorem}\label{tr3}
Let $\mathcal C$ be an $(n,d,w)$ constant-weight code with distance distribution $\{B_i\}_{i=0}^n$. For each $k = 1, 2, \ldots, n$,
\begin{eqnarray}\label{congthuc2k2}
\sum_{i=d/2}^w P_k^-(n;2i)B_{2i} \leq  \frac{2}{M} \Bigl[ \left( \left( n \atop k \right) - r_k\right)q_k(M-q_k) + r_k (q_k +1)(M-q_k - 1)\Bigr],
\end{eqnarray}
where $q_k$ and $r_k$ are the quotient and the remainder, respectively, when dividing $M P_k^-(n;w)$ by $\left(n \atop k \right)$, i.e.
\begin{eqnarray}
M P_k^-(n;w) = q_k \left(n \atop k \right) + r_k
\end{eqnarray}
with $0 \leq r_k < \left(n \atop k \right)$, and where $P^-_k(n;x)$ is defined by
\begin{eqnarray}
P_k^-(n;x) = \sum_{\substack{j=0\\j \mbox{\footnotesize{ odd }}}}^n \left(x \atop j \right) \left(n-x \atop k - j \right).
\end{eqnarray}

\end{theorem}

\subsubsection{New Linear Constraints on Distance Distributions $\{B_{i}\}_{i=0}^n$}
Linear constraints which correspond to those in Theorem \ref{kq2} have not been studied for constant-weight codes even though similar constraints have been studied by Argrell, Vardy, and Zeger in \cite{avz} (see Theorem \ref{tr5} below). We now present these constraints. Several new notations are needed. For convenience, we fix the following settings until the end of this section.
\begin{itemize}
\item $\mathcal C$ is an $(n,d,w)$ constant-weight code with distance distribution $\{B_i\}_{i=0}^n$ such that $d$ is even and $d < 2w \leq n$.
\item Let $v = n-w$. Since $2w \leq n$, $w \leq v$.
\item Let $H = \{d/2, d/2+1, \ldots, w\}$, which is the set of all positive integer $i$ such that $B_{2i}$ can be nonzero.
\item For each $i \in H$, let $\mathcal V_i$ be the set of all vectors $X$ in $\mathcal F^n$ such that $X$ has exactly $i$ ones on the first $w$ coordinates and exactly $i$ ones on the last $v = n-w$ coordinates.
\item For $i \not = j$ both in $H$, define
\begin{eqnarray}
m_{i,j} = \max\{d(X,Y) \mid X \in \mathcal V_i, Y \in \mathcal V_j\}.
\end{eqnarray}
\item For each codeword $X$ in $\mathcal C$, let
\begin{eqnarray}
S_{2i}(X) = \{Y \in \mathcal C \mid d(X,Y) = 2i\},
\end{eqnarray}
which is the set of all codewords $Y$ in $\mathcal C$ at distance $2i$ from $X$. By definition of $\{B_i\}_{i=0}^n$,
\begin{eqnarray}
B_{2i} = \frac{1}{|\mathcal C|} \sum_{X \in \mathcal C}|S_{2i}(X)|
\end{eqnarray}
for each $i \in H$.
\item For each $i \in H$, let $Q_i$ denote an integer such that
\begin{eqnarray}\label{qi}
T(i,w,i,v,d) \leq Q_i.
\end{eqnarray}
\item For $i \not = j$ both in $H$ with $i+j \geq v$ and $m_{i,j}=d$, let $Q_{ji}$ denote an integer such that
\begin{eqnarray}\label{a2}
T(w-j, i, v-j, i, d) \leq Q_{ji},
\end{eqnarray}
\end{itemize}

\begin{proposition}\label{tinhmij}
For $i \not = j$ both in $H$,
\begin{eqnarray}
m_{i,j} = a + b,
\end{eqnarray}
where
\begin{center}
$a = \left\{
\begin{array}{ll}
i + j & \mbox{ if } i + j < w\\
i + j - 2(i+j-w) & \mbox{ if } i + j \geq w
\end{array}
\right.$
\end{center}
and
\begin{center}
$b = \left\{
\begin{array}{ll}
i + j & \mbox{ if } i + j < v\\
i + j - 2(i+j-v) & \mbox{ if } i + j \geq v
\end{array}
\right..$
\end{center}
In particular, if $i+j \geq v \geq w$, then
\begin{eqnarray}
m_{i,j} = 2(n - i - j).
\end{eqnarray}
\end{proposition}

\begin{proof}
The proof is straightforward.
\end{proof}

\begin{lemma}\label{lmmoi15}
For each $i \in H$ and each codeword $X \in \mathcal C$,
\begin{eqnarray}\label{s2i}
|S_{2i}(X)| \leq Q_i.
\end{eqnarray}
\end{lemma}

\begin{proof}
Let $X$ be a codeword in $\mathcal C$. It is easy to see that $|S_{2i}(X)|$ is upper bounded by $A(n,\Lambda,d)$, where $\Lambda = \{(0,w),(X,2i)\}$. By Propositions \ref{mainp} and \ref{append} (iii), \begin{eqnarray}
A(n,\Lambda,d) \leq T(w-i, w, i, v,d) = T(i,w,i,v,d).
\end{eqnarray}
Hence, $|S_{2i}(X)| \leq T(i,w,i,v,d) \leq Q_i$.
\end{proof}

\begin{theorem}\label{dlmot1}
Suppose that $H_1$ is a nonempty subset of $H$ such that $m_{i,j} < d$ for all $i \not = j$ both in $H_1$. Then for each codeword $X \in \mathcal C$, $S_{2i}(X)$ is nonempty for at most one $i$ in $H_1$. Furthermore,
\begin{eqnarray}\label{claim1}
\sum_{i \in H_1} \frac{B_{2i}}{Q_i} \leq 1.
\end{eqnarray}
\end{theorem}

\begin{proof}
Let $X$ be a codeword in $\mathcal C$. Suppose on the contrary that there exist $i \not = j$ both in $H_1$ such that $S_{2i}(X)$ and $S_{2j}(X)$ are nonempty. Then choose any $Y \in S_{2i}(X)$ and $Z \in S_{2j}(X)$. By rearranging the coordinates, we may assume that
\begin{eqnarray}
\begin{array}{cccc}
X & = & \overbrace{1 \cdots 1}^w & \overbrace{0 \cdots 0}^v.
\end{array}
\end{eqnarray}
Since $d(X,Y) = 2i$ and $X$ and $Y$ have the same weight $w$, $Y+X$ must have exactly $i$ ones on the first $w$ coordinates and exactly $i$ ones on the last $v$ coordinates. This means $Y + X \in \mathcal V_i$. Similarly, $Z + X \in \mathcal V_j$. By definition of $m_{i,j}$, $d(Y+X, Z+X) \leq m_{i,j}$. Thus,
\begin{eqnarray}
d(Y,Z) = d(Y+X, Z+X) \leq m_{i,j} < d,
\end{eqnarray}
which is a contradiction since $Y$ and $Z$ are two different codewords in $\mathcal C$. Hence, $S_{2i}(X)$ is nonempty for at most one $i$ in $H_1$. It follows by Lemma \ref{lmmoi15} that
\begin{eqnarray}\label{mbehon}
\sum_{i \in H_1} \frac{|S_{2i}(X)|}{Q_i} \leq 1.
\end{eqnarray}
Taking sum of (\ref{mbehon}) over all $X \in \mathcal C$, we get
\begin{eqnarray}
\sum_{i \in H_1} \frac{B_{2i}}{Q_i} \leq 1.
\end{eqnarray}
\end{proof}

We now consider the case $m_{i,j} = d$ for some $i \not = j$ both in $H$. The following Lemma says that the existence of a codeword at distance $2i$ from $X$ may reduce the total number of codewords at distance $2j$ from $X$.

\begin{lemma}\label{lmQji}
Suppose $i \not = j$ both in $H$ such that $i + j \geq v$ and $m_{i,j} = d$. If $X$ is a codeword in $\mathcal C$ such that $|S_{2i}(X)| \geq 1$, then
\begin{eqnarray}
|S_{2j}(X)| \leq Q_{ji}.
\end{eqnarray}
\end{lemma}

\begin{proof}
Fix a codeword $Y \in S_{2i}(X)$. If $S_{2j}(X)$ is empty, then there is nothing to prove. Hence, we assume $|S_{2j}(X)| \geq 1$. Let $Z \in S_{2j}(X)$. By rearranging the coordinates, we may assume that
\begin{eqnarray}\label{aa}
\begin{array}{rcll}
X & = & \overbrace{1 \cdots 1}^w & \overbrace{0 \cdots 0}^v\\
\end{array}
\end{eqnarray}
As in the proof of Theorem \ref{dlmot1}, we can show that $Y+X \in \mathcal V_{i}$ and $Z+X \in \mathcal V_j$. By definition of $m_{i,j}$,
\begin{eqnarray}
d \leq d(Y,Z) = d(Y+X, Z+X) \leq m_{i,j} = d.
\end{eqnarray}
Thus,
\begin{eqnarray}
d(Y,Z) = d(Y+X,Z+X) = m_{i,j} = d.
\end{eqnarray}

Since $i + j \geq v \geq w$, by rearranging the first $w$ coordinates, we may assume that on the first $w$ coordinates:
\begin{eqnarray}\label{bb}
\begin{array}{rcccc}
Y + X & = & 1 \cdots 1 & 1 \cdots 1 & \overbrace{0 \cdots 0}^{w-i} \mid \cdots\\
Z + X & = & \underbrace{0 \cdots 0}_{w-j} & \underbrace{1 \cdots 1}_{i+j-w} & 1 \cdots 1\mid \cdots
\end{array}.
\end{eqnarray}
On the first $w$ coordinates, $Z + X$ must have exactly $i + j - w$ ones on the first $i$ coordinates (the other $w - i$ ones of $Z+X$ must be fixed since $d(Y+X,Z+X) = m_{i,j}$).

Similarly, since $i + j \geq v$, by rearranging the last $v$ coordinates, we may assume that on the last $v$ coordinates:
\begin{eqnarray}\label{cc}
\begin{array}{rccccc}
Y + X & = & \cdots \mid & 1 \cdots 1 & 1 \cdots 1 & \overbrace{0 \cdots 0}^{v-i}\\
Z + X & = & \cdots \mid & \underbrace{0 \cdots 0}_{v-j} & \underbrace{1 \cdots 1}_{i+j-v} & 1 \cdots 1
\end{array}.
\end{eqnarray}
On the last $v$ coordinates, $Z+X$ must have exactly $i+j-v$ ones on the first $i$ coordinates (the other $v - i$ ones of $Z+X$ must be fixed since $d(Y+X, Z+X) = m_{i,j}$).

From (\ref{aa}), (\ref{bb}), and (\ref{cc}), we get
\begin{eqnarray}
\nonumber d(Z,X+Y) & = & wt(X+Y+Z)\\
\nonumber & = & wt(X + (Y+X) + (Z+X))\\
\nonumber & = & (i+j-w) + (v-j + v-i)\\
& = & 2v - w.
\end{eqnarray}

Now the number of $Z \in S_{2j}(X)$ is upper bounded by $A(n,\Lambda,d)$, where $\Lambda = \{(0,w),(X,2j),(Y,d),(X+Y,2v-w)\}$. By Proposition \ref{tinhmij},
\begin{eqnarray}
d = m_{i,j} = 2(n-i-j).
\end{eqnarray}
Applying Proposition \ref{mainp2}, we get (by replacing $d = 2(n-i-j)$ and $n= w+v$)
\begin{eqnarray}
\nonumber A(n,\Lambda,d) & = & T(w-j,i,0,w-i,i+j-v,i,v-i,v-i,d)\\
& = & T(w-j,i,v-j,i,d),
\end{eqnarray}
where the last equality comes from Proposition \ref{append} in the appendix. Therefore,
\begin{eqnarray}\label{a1}
\nonumber |S_{2j}(X)| & \leq & A(n,\Lambda,d)\\
\nonumber & = & T(w-j,i,v-j,i,d) \\
& \leq & Q_{ji}.
\end{eqnarray}
\end{proof}

\begin{theorem}\label{tr4}
Suppose that $H_1$ is a subset of $H$ satisfying the following properties.
\begin{itemize}
\item $|H_1| \geq 2$.
\item There exist $i \not = j$ both in $H_1$ such that $i + j \geq v$ and $m_{i,j} = d$.
\item For all $k \not =  l$ both in $H_1$ such that either $k \not \in \{i,j\}$ or $l \not \in \{i,j\}$, we always have $m_{k,l} < d$.
\end{itemize}
Let $H_2 = H_1 \setminus \{i,j\}$. Then
\begin{eqnarray}\label{mot1}
\frac{Q_j - Q_{ji}}{Q_j Q_{ij}} B_{2i} + \frac{1}{Q_j}B_{2j} + \sum_{k \in H_2} \frac{1}{Q_k} B_{2k} \leq 1, \quad \mbox{ if } \frac{Q_{ij}}{Q_i} + \frac{Q_{ji}}{Q_j} \geq 1,
\end{eqnarray}
\begin{eqnarray}\label{hai2}
\frac{1}{Q_i}B_{2i} + \frac{Q_i - Q_{ij}}{Q_i Q_{ji}} B_{2j} + \sum_{k \in H_2} \frac{1}{Q_k} B_{2k} \leq 1, \quad \mbox{ if } \frac{Q_{ij}}{Q_i} + \frac{Q_{ji}}{Q_j} \geq 1,
\end{eqnarray}
\begin{eqnarray}\label{ba3}
\sum_{k \in H_1} \frac{1}{Q_k} B_{2k} \leq 1, \quad \mbox{ if } \frac{Q_{ij}}{Q_i} + \frac{Q_{ji}}{Q_j} \leq 1.
\end{eqnarray}
\end{theorem}

\begin{proof}
We first prove (\ref{mot1}). It suffices to show that for every codeword $X$ in $\mathcal C$,
\begin{eqnarray}\label{claimm}
\frac{Q_j - Q_{ji}}{Q_j Q_{ij}} |S_{2i}(X)| + \frac{1}{Q_j} |S_{2j}(X)| + \sum_{k \in H_2} \frac{1}{Q_k} |S_{2k}(X)| \leq 1,
\end{eqnarray}
if $\frac{Q_{ij}}{Q_i} + \frac{Q_{ji}}{Q_j} \geq 1$. Let $X$ be any codeword in $\mathcal C$. By Lemma \ref{lmmoi15},
\begin{eqnarray}
|S_{2i}(X)| \leq Q_i \quad \mbox{ and } \quad |S_{2j}(X)| \leq Q_j.
\end{eqnarray}
By Lemma \ref{lmQji},
\begin{eqnarray}
|S_{2i}(X)| \leq Q_{ij} \mbox{ if } |S_{2j}(X)| \geq 1,\\
|S_{2j}(X)| \leq Q_{ji} \mbox{ if } |S_{2i}(X)| \geq 1.
\end{eqnarray}
We prove (\ref{claimm}) by considering the following three cases.

{\em Case $1$: $|S_{2i}(X)| = 0$.} Proving (\ref{claimm}) is exactly the same as proving (\ref{mbehon}). So we are done.

{\em Case $2$: $|S_{2i}(X)| \geq 1$ and $|S_{2j}(X)| = 0$.} Since $|S_{2i}(X)| \geq 1$, $|S_{2k}(X)| = 0$ for every $k \in H_2$ by Theorem \ref{dlmot1}. Hence, to prove (\ref{claimm}), we only need to prove that
\begin{eqnarray}
(Q_j - Q_{ji}) |S_{2i}(X)| \leq Q_j Q_{ij}.
\end{eqnarray}
By hypothesis, $\frac{Q_{ij}}{Q_i} + \frac{Q_{ji}}{Q_j} \geq 1$. Thus, $(Q_j - Q_{ji}) Q_i \leq Q_j Q_{ij}$ and hence
\begin{eqnarray}
(Q_j - Q_{ji}) |S_{2i}(X)| \leq (Q_j - Q_{ji}) Q_i \leq Q_j Q_{ij}.
\end{eqnarray}

{\em Case $3$:} $|S_{2i}(c)|  \geq 1$ and $|S_{2j}(c)| \geq 1$. As in Case $2$, $|S_{2k}(X)| = 0$ for every $k \in H_2$. We have
\begin{eqnarray}
\nonumber \displaystyle \frac{Q_j - Q_{ji}}{Q_jQ_{ij}} |S_{2i}(X)| + \frac{1}{Q_j}|S_{2j}(X)| & \leq & \displaystyle \frac{Q_j - Q_{ji}}{Q_jQ_{ij}} Q_{ij} + \frac{1}{Q_j}Q_{ji} \\
\nonumber & = &\displaystyle 1- \frac{Q_{ji}}{Q_j} + \frac{Q_{ji}}{Q_j}\\
& = & 1.
\end{eqnarray}
Therefore, (\ref{claimm}) is proved and so is (\ref{mot1}).

By symmetry, (\ref{hai2}) follows.

We now prove (\ref{ba3}). It suffices to show that for every codeword $X$ in $\mathcal C$,
\begin{eqnarray}\label{claimmm}
\sum_{k \in H_1} \frac{1}{Q_k} |S_{2k}(X)| \leq 1,
\end{eqnarray}
if $\frac{Q_{ij}}{Q_i} + \frac{Q_{ji}}{Q_j} \leq 1$. If either $|S_{2i}(X)| = 0$ or $|S_{2j}(X)| = 0$, then proving (\ref{claimmm}) is exactly the same as proving (\ref{mbehon}). Hence, suppose that $|S_{2i}(X)| \geq 1$ and $|S_{2j}(X)| \geq 1$. As in Case $2$, $|S_{2k}(X)| = 0$ for every $k \in H_2$. We have
\begin{eqnarray}
\frac{1}{Q_i}|S_{2i}(X)| + \frac{1}{Q_j} |S_{2j}(X)| \leq \frac{1}{Q_i}Q_{ij} + \frac{1}{Q_j}Q_{ji} \leq 1.
\end{eqnarray}
\end{proof}

We now specify which $H_1$ are used in Theorems \ref{dlmot1} and \ref{tr4}. Let
\begin{eqnarray}
\alpha = d/2 - (n-2w)
\end{eqnarray}
and let
\begin{eqnarray}
\alpha_1 = \left\lfloor \frac{\alpha+1}{2} \right\rfloor \mbox{ and }\alpha_2 = \left\lfloor \frac{\alpha}{2} \right\rfloor
\end{eqnarray}
so that $\alpha_1 + \alpha_2 = \alpha$. Also, let
\begin{eqnarray}
i_0 = w - \alpha_1 \mbox{ and } j_0 = w - \alpha_2.
\end{eqnarray}

\begin{itemize}
\item {\em Case $1$: $\alpha$ is even.} In this case, $i_0 = j_0$. We apply Theorem \ref{dlmot1} for
    \begin{eqnarray}
    H_1 = \{j_0, j_0 + 1, \ldots, w\}
    \end{eqnarray}
    and apply Theorem \ref{tr4} for
    \begin{eqnarray}
    H_1 = \{i_0 - \epsilon, j_0 + \epsilon, j_0 + \epsilon + 1, \ldots, w\}
    \end{eqnarray}
    (with $i = i_0 - \epsilon$ and $j = j_0 + \epsilon$) for each $\epsilon = 1, 2, \cdots ,w-j_0$ .

\item {\em Case $2$: $\alpha$ is odd.} In this case, $i_0 < j_0$. We apply Theorem \ref{tr4} for
    \begin{eqnarray}
    H_1 = \{i_0 - \epsilon, j_0 + \epsilon, j_0 + \epsilon + 1, \ldots, w\}
    \end{eqnarray}
    (with $i = i_0 - \epsilon$ and $j = j_0 + \epsilon$) for each $\epsilon = 0,1, \cdots ,w-j_0$.
\end{itemize}

\begin{example}
Consider $(n,d,w)=(27,8,13)$. We have $\alpha = d/2 - (n - 2w) = 3$ is odd. Hence, $\alpha_1 = 2$ and $\alpha_2 = 1$. So, $i_0 = 11$ and $j_0 = 12$.
We can apply Theorem \ref{tr4} for $H_1 = \{i = i_0, j = j_0, w\} = \{11, 12, 13\}$ (with $\epsilon = 0$). We have
\begin{eqnarray}
\begin{array}{l}
\nonumber Q_i  =   26  \geq T(2,13, 3, 14, 8) = T(11,13,11,14,8),\\
\nonumber Q_j =  1  = T(1, 13, 2, 14, 8) = T(12,13,12,14,8),\\
\nonumber Q_{ij} =  20  \geq T(2, 12, 3, 12, 8),\\
\nonumber Q_{ji} = 1 = T(1, 11, 2, 11, 8),
\end{array}
\end{eqnarray}
and
\begin{eqnarray}
\begin{array}{l}
\nonumber Q_k = 1 = T(0, 13, 1, 14, 8) = T(13, 13, 13, 14, 8)
\end{array}
\end{eqnarray}
for $k = 13$. Since $\frac{Q_{ij}}{Q_i} + \frac{Q_{ji}}{Q_j} = \frac{20}{26} + \frac{1}{1} \geq 1$, Theorem \ref{tr4} gives
\begin{eqnarray}\label{lc1}
B_{24}+ B_{26} \leq 1
\end{eqnarray}
and
\begin{eqnarray}
\frac{1}{26} B_{22} + \frac{26 -20}{26} B_{24} + B_{26} \leq 1.
\end{eqnarray}
The later constraint is equivalent to
\begin{eqnarray}\label{lc2}
B_{22} + 6 B_{24} + 26 B_{26} \leq 26.
\end{eqnarray}
For $H_1=\{10, 13\}$ (with $\epsilon = 1$), Theorem \ref{tr4} gives less effective linear constraints.
\end{example}

\begin{table}[!t]
\renewcommand{\arraystretch}{1.3}
\caption{New upper bounds for $A(n,d,w)$}
\label{table_example2}
\centering
\begin{tabular}{|r|r|r|r|r|r|r|}
\hline
   & &  & \mbox{best}  & \mbox{best upper} &                 & \mbox{} \\
   & &  & \mbox{lower} & \mbox{bound}      & \mbox{new}      & \mbox{} \\
   & &  & \mbox{bound} & \mbox{previously} & \mbox{upper}    & \mbox{Schrijver} \\
n & d &w& \mbox{known} & \mbox{known}      & \mbox{bound}    & \mbox{bound} \\
\hline
\hline

20 & 6 & 8 & 588 & 1107 & 1106 & 1136\\
\hline
\hline
22 & 8 & 10 & 616 & 634 & 630 & 634\\

\hline
23 & 8 & 9 & 400 & 707 & 703  & 707\\

\hline
26 & 8 & 9  & 887  & 2108 & 2104 & 2108\\

26 & 8 & 11 & 1988 & 5225 & 5208 & 5225\\

\hline
27 & 8 & 9  & 1023 & 2914 & 2882 &  2918\\

27 & 8 & 11 & 2404 & 7833 & 7754 & 7833\\

27 & 8 & 12 & 3335 & 10547 & 10460 &  10697\\

27 & 8 & 13 & 4094 & 11981 & 11897 & 11981\\

\hline
28 & 8 &  9 & 1333 & 3895 & 3886 & 3900\\

28 & 8 & 11 & 3773 & 11939 & 11896 &  12025\\

28 & 8 & 12 & 4927 & 17011 & 17008 & 17011\\

28 & 8 & 13 & 6848 & 21152 & 21148 &  21152\\
\hline
\hline
23 & 10 & 9 & 45 & 81 & 79 & 82 \\

\hline
25 & 10 & 11 & 125 & 380 & 379 &380 \\

25 & 10 & 12 & 137 & 434 & 433 & 434 \\
\hline

26 & 10 & 11 & 168 & 566 & 565 & 566 \\

26 & 10 & 12 & 208 &  702 & 691 & 702\\
\hline
27 & 10 & 11 & 243 & 882 & 871 & 882\\

27 & 10 & 12 & 351 & 1201 & 1190 & 1201\\

27 & 10 & 13 & 405 & 1419 & 1406 & 1419\\
\hline
28 & 10 & 11 & 308 & 1356 & 1351 & 1356\\
\hline
\hline
25& 12 & 10 & 28 & 37 & 36 & 37\\
\hline
\end{tabular}
\end{table}

When $\alpha \leq 0$, there is no set $H_1$ satisfying Theorem \ref{tr4}. In this case, the following type of linear constraints which comes from \cite[Proposition 17]{avz} is very useful. As in \cite{avz}, let $T'(w_1, n_1, w_2, n_2, d)$ be the largest possible size of a $(w_1, n_1, w_2, n_2, d)$ doubly-bounded-weight code (a $(w_1, n_1, w_2, n_2, d)$ {\it doubly-bounded-weight code} is an $(n_1+n_2, d, w_1+w_2)$ constant-weight code such that every codeword has at most $w_1$ ones on the first $n_1$ coordinates). Tables for upper bounds on $T'(w_1,n_1,w_2,n_2,d)$ can be found on Erik Agrell's website http://webfiles.portal.chalmers.se/s2/research/kit/bounds/dbw.html.

\begin{theorem}\label{tr5}
Let $\delta = d/2$. For $i, j \in \{\delta, \delta+1, \ldots, w\}$ with $i \not = j$. If $i + j \leq n - \delta$, define $P_{ij}$ and $P_{ji}$ as any nonnegative integers such that
\begin{eqnarray}
P_{ij} \geq \min \{ P_i, T'(\Delta, j, i-\Delta, n-w-j, 2i-2\Delta\},\\
P_{ji} \geq \min \{ P_j, T'(\Delta, i, j-\Delta, n-w-i, 2j-2\Delta\},
\end{eqnarray}
where $\Delta := w - \delta$. Also, define $P_k := Q_k$ for each $k \in H$. Then
\begin{eqnarray}
P_{ji} B_{2i} + (P_i - P_{ij}) B_{2j} \leq P_i P_{ji}, \quad \mbox{ if } \frac{P_{ij}}{P_i} + \frac{P_{ji}}{P_j} > 1,\\
(P_j - P_{ji}) B_{2i} + P_{ij} B_{2j} \leq P_j P_{ij}, \quad \mbox{ if } \frac{P_{ij}}{P_i} + \frac{P_{ji}}{P_j} > 1,\\
P_j B_{2i} + P_i B_{2j} \leq P_i P_j, \quad \mbox{ if } \frac{P_{ij}}{P_i} + \frac{P_{ji}}{P_j} \leq 1.
\end{eqnarray}
\end{theorem}

By adding the linear constraints in Theorems \ref{tr1}, \ref{tr3}, \ref{dlmot1}, \ref{tr4}, and \ref{tr5} to Schrijver's semidefinite programming bound (\ref{maxcwc}), we obtained new upper bounds on $A(n,d,w)$ shown on Table \ref{table_example2}. As before, all computations were done by the same algorithm SDPT3 at the same server.

\appendix[Upper Bounds on $T(w_1, n_1, w_2, n_2, w_3, n_3, w_4, n_4, d)$]
To apply Theorem \ref{tr1}, we need tables of upper bounds on $T(w_1, n_1, w_2, n_2, w_3, n_3, w_4, n_4, d)$. However, there are no such tables available since this is the first time the function $T(w_1, n_1, w_2, n_2, w_3, n_3, w_4, n_4, d)$ is introduced. We show here some elementary properties that are used to obtain upper bounds on $T(w_1, n_1, w_2, n_2, w_3, n_3, w_4, n_4, d)$.

In general, let us define $T(\{(w_i,n_i)\}_{i=1}^t,d)$ as follows. For $t \geq 1$, a $(\{(w_i,n_i)\}_{i=1}^t,d)$ {\it multiply constant-weight code} is a $(\sum_{i=1}^t n_i, d)$ code such that there are exactly $w_i$ ones on the $n_i$ coordinates. When $t=1$ this is definition of an $(n_1, d, w_1)$ constant-weight code, when $t=2$ this is definition of a $(w_1, n_1, w_2, n_2, d)$ doubly-constant-weight code, etc.. Let $T(\{(w_i,n_i)\}_{i=1}^t,d)$ be the largest possible size of a $(\{(w_i,n_i)\}_{i=1}^t,d)$ multiply constant-weight code.

We present here elementary properties that are used to get upper bounds on $T(\{(w_i,n_i)\}_{i=1}^t,d)$. The proofs of these properties are similar to those for $A(n,d,w)$ or $T(w_1,n_1,w_2,n_2,d)$, and hence are omitted. Upper bounds on $T(w_1, n_1, w_2, n_2, w_3, n_3, w_4, n_4, d)$ that we used in Theorem \ref{tr1} are the best upper bounds obtained from these properties.

\begin{proposition}\label{append}
\begin{itemize}
\item [(i)] If $d$ is odd then,
\begin{eqnarray}
T(\{(w_i,n_i)\}_{i=1}^t,d) = T(\{(w_i,n_i)\}_{i=1}^t,d+1).
\end{eqnarray}

\item [(ii)] If $w_j=0$ for some $j \in \{1, 2, \ldots, t\}$, then
\begin{eqnarray}
T(\{(w_i,n_i)\}_{i=1}^t,d) = T(\{(w_i,n_i)\}_{i\not=j},d).
\end{eqnarray}

\item [(iii)] $T(\{(w_i,n_i)\}_{i=1}^t,d)$ does not change if we replace any $w_i$ by $n_i - w_i$.

\item [(iv)] $T(\{(w_i,n_i)\}_{i=1}^t,2) = \prod_{i=1}^t \left(n_i \atop w_i \right)$.

\item [(v)] $T(\{(w_i,n_i)\}_{i=1}^t, 2\sum_{i=1}^t w_i) = \min_{1 \leq i \leq t} \left\lfloor \frac{n_i}{w_i} \right \rfloor$.

\item [(vi)] $T(\{(w_i,n_i)\}_{i=1}^t,d) = 1$ if $2 \sum_{i=1}^t w_i < d$.
\end{itemize}
\end{proposition}

\begin{remark}
By (i) and (iv), we can always assume that $d$ is even and $d \geq 4$. By (ii) and (iii), we may assume that $0<2w_i\leq n_i$ for each $i$.
Also, by (v) and (vi), we can assume that $d < 2\sum_{i=1}^t w_i$.
\end{remark}

The next proposition can be used to reduce the size of $\{(w_i,n_i)\}_{i=1}^t$ from $t$ to $t-1$. When the size of the set is $2$, we use known upper bounds on $T(w_1, n_1, w_2, n_2, d)$.

\begin{proposition} If $t \geq 2$, then
\begin{eqnarray}
T(\{(w_i,n_i)\}_{i=1}^t,d) \leq T(\{(w'_i,n'_i)\}_{i=1}^{t-1},d),
\end{eqnarray}
where $w'_i = w_i, n'_i = n_i$ for $i = 1, 2, \ldots, t-2$, and $w'_{t-1} = w_{t-1} + w_t, n'_{t-1} = n_{t-1} + n_t$.
\end{proposition}

\begin{proposition}
If $w_i >0$, then
\begin{eqnarray}
T(\{(w_i,n_i)\}_{i=1}^t,d) \leq \left \lfloor \frac{n_i}{w_i} T(\{(w'_i,n'_i)\}_{i=1}^t,d) \right \rfloor,
\end{eqnarray}
where $\{(w'_i,n'_i)\}_{i=1}^t$ is obtained from $\{(w_i,n_i)\}_{i=1}^t$ by replacing the pair $(w_i,n_i)$ by $(w_i-1, n_i -1)$.
\end{proposition}

\begin{proposition}
If $w_i < n_i$, then
\begin{eqnarray}
T(\{(w_i,n_i)\}_{i=1}^t,d) \leq \left \lfloor \frac{n_i}{n_i-w_i} T(\{(w'_i,n'_i)\}_{i=1}^t,d) \right \rfloor,
\end{eqnarray}
where $\{(w'_i,n'_i)\}_{i=1}^t$ is obtained from $\{(w_i,n_i)\}_{i=1}^t$ by replacing the pair $(w_i,n_i)$ by $(w_i, n_i -1)$.
\end{proposition}


%





\ifCLASSOPTIONcaptionsoff
  \newpage
\fi

\end{document}